%% file: main.tex
\documentclass[a4paper]{article}
\usepackage[utf8]{inputenc}
\usepackage{physics}
\usepackage{a4wide}
\usepackage[export]{adjustbox}
\usepackage{tabularx}
\usepackage{float}
\usepackage{cite}
\usepackage{subfig}
\usepackage{multirow}
\graphicspath{{Figures/}{../Figures/}}
\usepackage{subfiles}
\usepackage{wrapfig}
\usepackage{amssymb}
\usepackage{appendix}
\usepackage{amsmath}
\usepackage{mathtools}
\usepackage{xcolor}
\usepackage{qcircuit}
\usepackage[pdftex,colorlinks=true,linkcolor=black,citecolor=blue,urlcolor=black]{hyperref}
\usepackage{thm-restate}
\usepackage{enumitem} 
\usepackage{authblk}

\input{commands}

\title{Complexity of the Guided Local Hamiltonian Problem: Improved Parameters and Extension to Excited States}
\author[1,3]{Chris Cade}

\author[2]{Marten Folkertsma}

\author[2]{Jordi Weggemans}

\affil[1]{QuSoft \& University of Amsterdam (UvA), Amsterdam, the Netherlands}
\affil[2]{QuSoft \& CWI, Amsterdam, the Netherlands}
\affil[3]{Fermioniq, Amsterdam, the Netherlands}

\date{\today}

\begin{document}
\maketitle

\begin{abstract}
\noindent Recently it was shown that the so-called \emph{guided local Hamiltonian} problem -- estimating the smallest eigenvalue of a $k$-local Hamiltonian when provided with a description of a quantum state (`guiding state') that is guaranteed to have substantial overlap with the true groundstate -- is $\BQP$-complete for $k\geq 6$ when the required precision is inverse polynomial in the system size $n$, and remains hard even when the overlap of the guiding state with the groundstate is close to a constant ($\frac12 -\Omega( \frac{1}{\poly(n)}$)). 

We improve upon this result in three ways: by showing that it remains $\BQP$-complete when i) the Hamiltonian is 2-local, ii) the overlap between the guiding state and target eigenstate is as large as $1 - \Omega(\frac{1}{\poly(n)})$, and iii) when one is interested in estimating energies of excited states, rather than just the groundstate. Interestingly, iii) is only made possible by first showing that ii) holds. 
\end{abstract}

\input{introduction}
\input{GandGconstruction}

\input{2local}

\input{excited}

\input{containment}
\input{conclusion}

\section*{Acknowledgements}
We are grateful to Sevag Gharibian, Ryu Hayakawa, Fran{\c{c}}ois Le Gall and Tomoyuki Morimae for sharing their manuscript. We also thank Jonas Helsen for feedback on an earlier draft, and Ronald de Wolf for helpful comments. CC acknowledges support from QuSoft and CWI, as well as the University of Amsterdam (UvA) under a POC (proof of concept) fund. MF and JW were supported by the Dutch Ministry of Economic Affairs and Climate Policy (EZK), as part of the Quantum Delta NL programme.
\bibliography{main.bib}
\bibliographystyle{alpha}

\input{appendix}

\end{document}

%% file: commands.tex
\usepackage{amsthm}
\usepackage{thmtools,thm-restate}
\newtheorem{theorem}{Theorem}

\newtheorem{lemma}{Lemma}
\newtheorem{corollary}{Corollary}
\newtheorem{definition}{Definition}

\newtheorem{fact}{Fact}

\newcommand{\mO}{\mathcal{O}}

\newcommand{\mS}{\mathcal{S}}
\newcommand{\cS}{\mS}

\newcommand{\QMA}{\mathsf{QMA}}
\newcommand{\BQP}{\mathsf{BQP}}

\newcommand{\QCMA}{\mathsf{QCMA}}
\newcommand{\NP}{\mathsf{NP}}
\newcommand{\MA}{\mathsf{MA}}

\newcommand{\poly}{\mathop{poly}}

\newcommand{\kb}[2]{|#1\rangle\langle#2|}
\newcommand{\proj}[1]{|#1\rangle\langle#1|}
\newcommand{\ot}{\otimes}

\newcommand{\hatt}[1][]{\hat{t#1}}

\newcommand{\glh}[0]{\mathsf{GLHLE}(k,c,\delta,\zeta)}

%% file: introduction.tex
\section{Introduction}
Quantum chemistry is generally regarded as one of the most promising applications of quantum computers~\cite{Aaronson2009ComputationalCW,Bauer2020chemical}. To obtain information about the physical and chemical properties of quantum systems, one usually needs to estimate their spectral properties. For example, in order to understand chemical reaction processes, it is often necessary to know the relative energies of intermediate states along a particular reaction pathway. By comparing these energies, one can deduce which sequence of molecular transformations is most likely to occur in reality. Such energies are usually estimated via some computational method from the \textit{electronic structure Hamiltonian} associated to the system. The accuracy to which these energies are known is tremendously important: typically, in order to confidently distinguish between several reaction mechanisms, one would like to have an accuracy that is smaller than the so-called \textit{chemical accuracy}, which is about 1.6 millihartree.\footnote{This quantity, which is $\approx$1 kcal/mol, is chosen to match the accuracy achieved by thermochemical experiments.} Since in chemistry the norm of the Hamiltonian is allowed to scale polynomially in the number of particles and local dimension per particle, obtaining chemical accuracy corresponds to obtaining inverse polynomial precision when one considers (sub-)normalized Hamiltonians.\footnote{By normalized Hamiltonian, we mean a Hamiltonian $H$ such that $\|H\|\leq1$, where $\|\cdot\|$ denotes the operator norm.} 

The problem of estimating ground- and excited-state energies of the electronic structure Hamiltonian, without any additional information to help us, is known to be $\QMA$-hard~\cite{o2021electronic}.\footnote{$\QMA$ is the set of problems for which a {\sc yes}-instance can be verified efficiently with a quantum computer, and can be thought of as the quantum analogue of the class $\mathsf{NP}$.} Assuming that $\QMA \neq \BQP$, this suggests that estimating energies of physical Hamiltonians is a hard problem even for quantum computers. Hence, to better understand when and how quantum computers might provide significant speedups for computational problems in chemistry, it can be instructive to narrow down which problems are easy for quantum computers (i.e. inside $\BQP$), but seemingly difficult for classical ones. One way to do this is to search for problems that are $\BQP$-complete, implying that they cannot be solved (in polynomial time) on a classical computer unless $\BQP = \mathsf{BPP}$. 

Recently, Gharibian and Le Gall~\cite{gharibian2021dequantizing} raised and formalized the following natural question: \emph{``If we are given a (quantum) state guaranteed to be a good approximation to the true groundstate of a particular Hamiltonian, how difficult is it to accurately estimate the groundstate energy?"}. Such a question is motivated by the observation that, in practice, one often knows additional information that can help to compute energies (for example a state known to have \textit{energy} close to the ground energy, such as a Hartree-Fock state), which could make the task somewhat easier. 

\subsection{Definitions}
\paragraph{Notation} We denote by $[M]$ the set $\{1,\dots,M\}$. We write $\lambda_i(A)$ to denote the $i$th eigenvalue of a Hermitian matrix $A$, ordered in non-decreasing order, with $\lambda_0(A)$ denoting the smallest eigenvalue (ground energy).  We denote $\text{eig}(A) = \{ \lambda_0(A),\dots,\lambda_{\text{dim}(A)-1}(A)\}$ for the (ordered) set of all eigenvalues of $A$. For some Hilbert space $\mathcal{H} = \mathcal{S}_1 + \mathcal{S}_2$, we denote $A | _{\mathcal{S}_1}$ for restriction of $A$ to $\mathcal{S}_1$.

\noindent As mentioned above, the class of all computational \textit{promise problems} that can be efficiently solved by a quantum computer is called $\BQP$, of which the formal definition is listed below.
\begin{definition}[$\BQP$] $\BQP$ is the set of all languages $L = \{L_{\text{yes}},L_{\text{no}}\}  \subset \{ 0,1\}^*$ for which there exists a (uniform family of) quantum circuit $V$ of size $T=\poly(n)$ acting on $r=\poly(n)$ qubits such that for every input $x\in L$ of length $n=|x|$,
\begin{itemize}
    \item if $x \in L_{\text{yes}} $ then the probability that $V$ accepts input $\ket{x,0}$ is $\geq c(=2/3)$,
    \item if $x \in L_{\text{no}}$ then the probability that $V$ accepts input $\ket{x,0}$ is $\leq s(=1/3)$.
\end{itemize}
\end{definition}
\begin{fact} [Error reduction] the completeness and soundness parameters in $\BQP$ can be made exponentially close to $1$ and $0$, respectively, i.e. $c=1-2^{-\mO(n)}$ and $s=2^{-\mO(n)}$.
\label{fact:er}
\end{fact}
One can formulate the question mentioned in the preceding section as a decision problem, defined as the guided Local Hamiltonian problem ($\mathsf{GLH}$), recently introduced by Gharibian and Le Gall~\cite{gharibian2021dequantizing}. We generalize their formulation of the problem by considering arbitrary eigenstates, which we will denote as the \textit{Guided Local Hamiltonian Low Energy}-problem ($\mathsf{GLHLE}$). For this, we first need a definition of semi-classical states, which is used in the problem definition of $\mathsf{GLHLE}$.
\begin{definition}[Semi-classical state - from~\cite{grilo2015qma}] A quantum state $\ket{\psi}$ is semi-classical whenever it can be written as  
\begin{align*}
    \ket{\psi} = \frac{1}{\sqrt{|S|}} \sum_{x \in S} \ket{x},
\end{align*}
for any non-empty subset $S \subseteq \{0,1\}^n$ with $|S|=\mO(\poly(n))$ .
\end{definition}
Then the problem we consider in this paper is\footnote{Note that the promise on the guiding state in our definition is in terms of the fidelity instead of the overlap. Since the fidelity is the overlap squared, the results for both definitions are directly related.}:
\begin{definition} [Guided Local Hamiltonian Low Energy] $\glh$ \\
\textbf{Input:} A $k$-local Hamiltonian $H$ with $\|H\|\leq 1$ acting on $n$ qubits, and a semi-classical quantum state $u \in \mathbb{C}^{2^n}$, threshold parameters $a,b \in \mathbb{R}$ such that $b-a \geq \delta > 0$.\\
\textbf{Promises :} $\norm{\Pi_{c} u}^2 \geq \zeta$, where $\Pi_{c}$ denotes the projection on the subspace spanned by the $c$th eigenstate, ordered in order of non-decreasing energy, of $H$, and either $\lambda_c(H) \leq a$ or $\lambda_c(H) \geq b$ holds.\\
\textbf{Output:} \begin{itemize}
    \item If $\lambda_c(H) \leq a$, output {\sc yes}.
    \item If $\lambda_c(H) \geq b$, output {\sc no}.
\end{itemize}
\end{definition}

\subsection{Results}
In~\cite{gharibian2021dequantizing} the authors show that $\glh$ in the ground state setting (i.e. with $c=0$) is $\BQP$-complete for $k \geq 6$ and $\zeta \in (1/\poly,1/2-\Omega(1/\poly(n))$, and $\delta = 1/\poly(n))$~\cite{gharibian2021dequantizing}. In their construction, they only consider guiding states which are of the form of so-called \textit{semi-classical states}~\cite{grilo2015qma}.\footnote{Since these states are a class of quantum states to which sampling access (being able to compute individual amplitudes of computational basis states, as well as sample according to the squares of the amplitudes) can be efficiently provided classically, the problem defined with such guiding states allows a direct and fair comparison between classical and quantum algorithms.}

\noindent In this work, we generalize and strengthen their results by showing that the problem remains $\BQP$-complete for 2-local Hamiltonians, when the fidelity of the guiding state with the true groundstate is substantially larger, and when one considers eigenstates above the groundstate, whilst still considering semi-classical guiding states in the construction. To be precise, in this paper we prove the following:

\begin{theorem}[$\BQP$-hardness of $\mathsf{GLHLE}$]\label{thm:main_result}\label{thm:main_result1}
$\mathsf{GLHLE}(k,c,\zeta,\delta)$ is $\BQP$-hard for $k \geq 2$, $0 \leq c \leq \mO(\poly(n))$, $\zeta = \mO(1-1/\poly(n))$, and $\delta = 1/\Omega(\poly(n))$.
\end{theorem}

\begin{restatable}[Containment in $\BQP$ of $\mathsf{GLHLE}$]{theorem}{bqpcontainment}
\label{thm:main_result2}
~\
\begin{enumerate}[label=(\roman*)]
\item $\mathsf{GLHLE}(k,0,\zeta,\delta)$ is contained in $\BQP$ for $k=\mO(\log(n))$, $\zeta = \Omega(1/\poly(n))$, and $\delta = 1/\mO(\poly(n))$. 
\item $\mathsf{GLHLE}(k,c,\zeta,\delta)$ for $c \geq 1$ is contained in $\BQP$ when $k=\mO(\log(n))$, $\zeta = \frac12 + \Omega(1/\poly(n))$, and $\delta = 1/\mO(\poly(n))$. 
\end{enumerate}
\end{restatable}

The above theorems follow straightforwardly from the results of Sections~\ref{sec:2local},~\ref{sec:excited}, and~\ref{sec:containment}. The reason for the separation of Theorem~\ref{thm:main_result2} into parts (i) and (ii) is as follows: when the fidelity of the guiding state with the target eigenstate is sufficiently above $1/2$ (in this case $1/2 + 1/\poly(n)$), then by inputting this state to quantum phase estimation and measuring, one can choose the most frequently observed output to be the estimate of the energy of the target state (since we know that the fidelity with any other eigenstate will be smaller than the fidelity with the target state). On the other hand, if the fidelity is not sufficiently above $1/2$, then it might be the case that the guiding state has significant fidelity with other eigenstates, and that the energies of these states will be measured with equal or higher probability than that of the target eigenstate. In this case, it is impossible to decide (in polynomial time) which energy corresponds to the target state, and which to the other, unwanted states, unless the target state is the groundstate (case (i)), in which case we can employ the variational principle and simply choose the smallest energy. We leave as an open problem the containment in $\BQP$ of the case $c>0$ and $\zeta = \Omega(1/\poly(n))$.

\

\noindent Put together, we obtain the following result. Note that completeness for the case $c>0$ is only made possible due to the fact that the problem remains $\BQP$-hard for a fidelity as large as $1-1/\poly(n)$ (i.e. Theorem~\ref{thm:main_result1}).
\begin{theorem}[$\BQP$-completeness of $\mathsf{GLHLE}$]
$\glh$ is $\BQP$-complete for $k=2$, $c=\mO(\poly(n))$, $\frac12 + 1/\mO(\poly(n)) \leq \zeta \leq 1 - 1/\mO(\poly(n))$, and $\delta = 1/\poly(n)$.
\end{theorem}


\subsection{Simultaneous \& subsequent works} 
\paragraph{Simultaneous work} During preparation of this manuscript, we became aware of a parallel work by Le Gall, Gharibian, Hayakawa, and Morimae~\cite{Gharibian2022improved} which obtains similar results to ours. In particular, they also improve on the work of~\cite{gharibian2021dequantizing} by showing that the guided local Hamiltonian problem remains $\BQP$-hard for $2$-local Hamiltonians, and when the fidelity between the guiding state and groundstate is as large as $1-1/\mO(\poly(n))$, making use of the Schrieffer-Wolf transformation framework~\cite{bravyi2011schrieffer}. They also provide an alternative proof that is similar to ours, based on the projection lemma of~\cite{kempe2006} (Appendix C in~\cite{Gharibian2022improved}). In addition, they also prove that the problem remains hard for a family of physically motivated Hamiltonians, which includes the $XY$ model and the Heisenberg model. They do not consider the generalization to excited states that we do in this work, and hence the results of both papers provide additional results that are complimentary to each other.\footnote{One might say that our results and theirs have fidelity $2/3$.} A merged version of both works has since appeared, and we refer the reader to that paper \cite{Gharibian2022improved}.

\paragraph{Subsequent work} In~\cite{weggemans2023guidable} we consider variants of $\mathsf{GLH}$ in which the guiding state is no longer part of the \textit{input} to the problem, but merely \textit{promised} to exist and satisfy certain constraints (which are somewhat more natural than those considered in this work). We show how for a certain class of guiding states the corresponding guidable local Hamiltonian problem is $\QCMA$-complete in the inverse polynomial precision regime, but is in $\NP$ when the promise-gap is constant, reminiscent of the results that were obtained for the stoquastic frustration-free local Hamiltonian problem in~\cite{Aharonov2019stoquastic} (in which similar promise-gap scaling regimes determine whether the problem is $\MA$-complete or in $\NP$). Finally, we discuss the implications of these results in relation to the quantum PCP conjecture and heuristic ansätze state preparation. 

%% file: GandGconstruction.tex
\section{Ghariban and Le Gall's construction}\label{sec:legall_construction}
In this section we will briefly restate Ghariban and Le Gall's original construction~\cite{gharibian2021dequantizing}. Let $\Pi = (\Pi_{\text{yes}}, \Pi_{\text{no}})$ be a promise problem in $\BQP$, and $x\in \{0,1\}^n$ an input. Let $U= U_T \dots U_1$ be a poly-time uniformly generated quantum circuit consisting of $1$- and $2$-qubit gates $U_i$, deciding on $\Pi$. More precisely, $U$ takes an $n$-qubit input register $A$, and a $r=poly(n)$-qubit work register $B$ and outputs, upon measurement, a 1 on the first qubit with probability at least $\alpha$ (resp. at most $\beta$) if $x \in \Pi_{\text{yes}}$ (resp. $x \in \Pi_{\text{no}}$). By Fact~\ref{fact:er}, we can assume wlog $\alpha = 1 - 2^{-n}$,  $\beta = 2^{-n}$. 

Consider Kitaev's original 5-local clock Hamiltonian~\cite{Kitaev2002ClassicalAQ}, where $C$ denotes the `clock' register consisting of $T=\poly(n)$ qubits:
\begin{align}
    H_{\text{in}} &:= (I- \kb{x}{x})_A \ot (I - \kb{0 \dots 0}{0 \dots 0})_B \ot \kb{0}{0}_C, \nonumber\\
    H_{\text{out}} &:= \kb{0}{0}_{\text{out}} \ot \kb{T}{T}_C, \nonumber\\
    H_{\text{clock}} &:= \sum_{j = 1}^{T} \kb{0}{0}_{C_{j}} \ot \kb{1}{1}_{C_{j+1}}, \nonumber\\
    H_{\text{prop}} &:= \sum_{t = 1}^T H_t \qquad \text{where} \qquad \nonumber\\
    &\qquad H_t := -\frac{1}{2} U_t \ot \kb{t}{t-1}_C -\frac{1}{2} U_t^\dagger \ot \kb{t-1}{t}_C + \frac{1}{2} I \ot \kb{t}{t}_C + \frac{1}{2} I \ot \kb{t-1}{t-1}_C. \label{eq:propagation_hamiltonian}
\end{align}
The ground state energy of the Hamiltonian $H^{(5)} = H_\text{in} + H_\text{out} + H_\text{clock} + H_\text{prop}$ has the following property:
\begin{itemize}
    \item If $U$ accepts with at least probability $\alpha$, then $\lambda_{0} (H^{(5)}) \leq \frac{1-\alpha}{T}$.
    \item If $U$ accepts with at most probability $\beta$, then $\lambda_{0} (H^{(5)}) \geq \Omega(\frac{1-\beta}{T^3})$.
\end{itemize}
$H^{(5)}$ can be split into two separate terms:
\begin{align*}
    H_1 &:=  H_{\text{in}} +  H_{\text{clock}} +  H_{\text{prop}} \\
    H_2 &:= H_{\text{out}}
\end{align*}
such that the history state
\begin{align}
    \ket{\eta} = \frac{1}{\sqrt{T}}\sum_{t=1}^{T} U_t \dots U_1 \ket{x}_A \ket{0\dots0}_B \ket{t}_C,
    \label{eq:hist_state}
\end{align}
spans the null space $\mathcal S$ of $H_1$.  Consider the following semi-classical guiding state  in~\cite{gharibian2021dequantizing}
\[
    \ket{u} = \frac{1}{\sqrt{T}}\sum_{t = 1}^N \ket{x}_A\ket{0\dots 0}_B\ket{t}_C.
\]
In general this guiding state has at most  $\mO(1/(TN))$ fidelity with the history state and therefore an even smaller fidelity with the actual ground state of $H^{(5)}$. In order to meet the promise that $\norm{\Pi_{H'}\ket{u}}^2 \geq \zeta$ in both the {\sc yes}- and {\sc no}-case, Gharibian and Le Gall use the following tricks:
\begin{itemize}
    \item Since the history state in Eq.~\eqref{eq:hist_state} uniquely spans the null space of $H_1$ (of which all terms are positive semi-definite) in the $\BQP$-setting, and a bound is known on the energy of all non-zero eigenstates, then weighing $H_1$ with a large (but only polynomial) prefactor $\Delta$ allows one to increase the fidelity of the actual ground state with the history state.
    \item By \textit{pre-idling} the circuit $U$ that is in the clock Hamiltonian  -- i.e. applying $M$ identity gates before the first actual gate -- the fidelity between $\ket{u}$ and the history state can be increased. This increases the number of gates from $T$ to $T'=T+M$. Denote the weighted and pre-idled Hamiltonian as $\hat{H}^{(5)}$. Also, define 
    \begin{align*}
        \hat{\alpha} : = \frac{1-\alpha}{T'+1} \quad \text{and} \quad \hat{\beta}:=\Omega\left(\frac{1-\sqrt{\beta}}{{T'}^3}\right).
    \end{align*}
    \item Finally, by block-encoding $\hat{H}^{(5)}$ into a larger Hamiltonian $H^{(6)}$, which acts on $n+r+T+1$ qubits (adding another single-qubit register $D$), one can add another Hamiltonian (in their case a scaled identity term) in another block such that the ground space in case of a {\sc no}-case is trivial, only increasing the locality of the Hamiltonian in the construction by 1. By setting this specific qubit in the guiding state to the $\ket{+}$ state, one ensures that it has fidelity with both the {\sc no}- and {\sc yes}-cases.
\end{itemize}

The final Hamiltonian is then
\begin{align}
    H^{(6)}_{ABCD} := \frac{\hat{\alpha}+\hat{\beta}}{2} I_{ABC} \ot \kb{0}{0}_D + \hat{H}^{(5)}_{ABC} \ot \kb{1}{1}_D,
\end{align}
where $\hat{H}^{(5)}_{ABC} = \Delta \left( H_{\text{in}} +  H_{\text{clock}} +  H_{\text{prop}}   \right) + H_{\text{out}}$. The guiding state becomes
\begin{align}
    \ket{u} := \ket{x}_A \ket{0 \dots 0}_B \left( \frac{1}{\sqrt{M}} \sum_{t=1}^M \ket{t} \right)_C \ket{+}_D.
\end{align}
Since the overall construction starts from a $5$-local Hamiltonian, the last trick increases the locality to $6$ and restricts the fidelity to be at most $1/2-\Omega(1/\poly(n))$.

%% file: 2local.tex
\section{A 2-local construction with increased fidelity}\label{sec:2local}
To improve Ghariban and Le Gall's results in terms of the locality and fidelity parameter range, we make the following two observations:
\begin{itemize}
    \item Rather than using Kitaev's 5-local circuit-to-Hamiltonian, one can instead adapt the construction of a 2-local Hamiltonian introduced by Kempe, Kiteav, and Regev~\cite{kempe2006}. We modify their construction by removing the witness register (since we care about acceptance only for a fixed input $x$), and by blending it with the aforementioned tricks of pre-idling and gap amplification. We also prove a lower bound on the spectral gap of that Hamiltonian, which was not needed in~\cite{kempe2006} but is needed for our application.
    \item One can show that the gap amplification arising from the scaling  of $H_1$ increases not just the fidelity of the ground state with the history state in the {\sc yes}-case, but also in the {\sc no}-case. This observation allows one to circumvent the splitting of both cases into different blocks of a larger Hamiltonian, preserving the locality of the initial construction. 
\end{itemize}
We begin by briefly reviewing the 2-local construction from~\cite{kempe2006}. To turn the usual clock Hamiltonian into a 2-local one,~\cite{kempe2006} first assume that the initial $T$-gate circuit $U$ is composed only of single-qubit gates and the controlled-$Z$ gate $CZ$,\footnote{This is wlog since single qubit gates combined with $CZ$ are universal for quantum computation.} and that each $CZ$ gate is conjugated by single-qubit $Z$ gates acting on each qubit\footnote{Again wlog, since these gates commute with $CZ$ and therefore cancel to the identity.}, and finally that the $CZ$ gates are applied only at regular intervals. We will make the same assumption, except combine it with pre-idling from~\cite{gharibian2021dequantizing}. Precisely, this means that we construct a new circuit $V := U I_{M} \dots I_1$ from $U$ consisting of $T+M$ gates, where the first $M$ gates are the identity, the $CZ$ gates are applied only at time-steps $L+M, 2L+M, \dots, T_2L+M$ for some integer $L < M$ and with $T_2$ the total number of such gates, and the rest of the (non-identity) gates are single-qubit. Let $T_1 = \{1,\dots,T+M\}\setminus\{L+M,2L+M,\dots,T_2L+M\}$ be the times at which single qubit gates are applied (including the pre-idling identity gates). Then following~\cite{kempe2006}, our Hamiltonian is 
\begin{align}
      H = J_{\text{in}} H_{\text{in}} + J_{\text{clock}} H_{\text{clock}} + J_1 H_{\text{prop1}} + J_2 H_{\text{prop2}} + (T+M)H_{\text{out}}
      \label{eq:H_2local}
\end{align}
with $H_{\text{in}}$, $H_{\text{clock}}$, $H_{\text{out}}$ defined as in Section~\ref{sec:legall_construction}, and $J_{\text{in}}, J_{\text{clock}}, J_1, J_2$ coefficients to be chosen later. The two propagation terms are defined differently as 
\[
    H_{\text{prop1}} = \sum_{t \in T1} H_{\text{prop},t} \qquad\qquad H_{\text{prop2}} = \sum_{l=1}^{T_2} (H_{\text{qubit},lL+M} + H_{\text{time},lL+M})
\]
with 
\[
    H_{\text{prop},t} = \frac12 \left(I \ot \proj{10}_{t,t+1} + I \ot \proj{10}_{t-1,t} - U_t \ot \ketbra{1}{0}_t - U_t^\dagger \ot \ketbra{0}{1}_t \right)
\]
for $t \in T_1 \cap \{2,\dots,T-1\}$ and
\[
    H_{\text{prop,1}} = \frac12 \left(I \ot \proj{10}_{1,2} + I \ot \proj{0}_{1} - U_1 \ot \ketbra{1}{0}_1 - U_1^\dagger \ot \ketbra{0}{1}_1 \right)
\]
\[
    H_{\text{prop,T}} = \frac12 \left(I \ot \proj{1}_{T} + I \ot \proj{10}_{T-1,T} - U_T \ot \ketbra{1}{0}_T - U_T^\dagger \ot \ketbra{0}{1}_T \right)\,.
\]
Let $f_t$ and $s_t$ be the first and second qubit acted on by the gate $CZ$ at time $t$, and define
\[
    H_{\text{qubit},t} = \frac12 \left(-2\proj{0}_{f_t} -2\proj{0}_{s_t} + \proj{1}_{f_t} + \proj{1}_{s_t} \right) \ot (\ketbra{1}{0}_t + \ketbra{0}{1}_t)
\]
\begin{eqnarray*}
    H_{\text{time},t} = \frac18 I \ot &(&\proj{10}_{t,t+1} + 6\proj{10}_{t+1,t+2} + \proj{10}_{t+2,t+3}  \\
    &+& 2\ketbra{11}{00}_{t+1,t+2} + 2\kb{00}{11}_{t+1,t+2} \\
    &+& \kb{1}{0}_{t+1} + \kb{0}{1}_{t+1} + \kb{1}{0}_{t+2} + \kb{0}{1}_{t+2} \\
    &+& \proj{10}_{t-3,t-2} + 6\proj{10}_{t-2,t-1} + \proj{10}_{t-1,t} \\
    &+& 2\kb{11}{00}_{t-2,t-1} + 2\kb{00}{11}_{t-2,t-1} \\
    &+& \kb{1}{0}_{t-2} + \kb{0}{1}_{t-2} + \kb{1}{0}_{t-1} + \kb{0}{1}_{t-1}
    ) .
\end{eqnarray*}
Finally, let $\ket{\hat{t}}$ denote the valid clock state for time $t$, i.e.
\begin{equation}\label{eq:unary_clock}
    \ket{\hat{t}} :=  |\underbrace{1\dots1}_{t} \underbrace{0\dots0}_{(T+M)-t} \rangle\,.
\end{equation}
These are the states not given an energy penalty by $H_{\text{clock}}$, and as such span its nullspace (we touch on this further below).

\ 
Before we begin, we will need the following lemma from~\cite{kempe2006}, which allows us to bound (from above and below) the smallest eigenvalue of a Hamiltonian of the form $H = H_1 + H_2$ using knowledge of the smallest eigenvalue of $H_1$ restricted to the nullspace of $H_2$. The final claim in the lemma is not in the original statement from~\cite{kempe2006}, but follows trivially from their proof, which we include for completeness.
\begin{lemma}[`Projection lemma' -- Lemma 1 from~\cite{kempe2006}]\label{lem:projection}
    Let $H = H_1 + H_2$ be the sum of two Hamiltonians acting on some Hilbert space $\mathcal{H} = \mathcal{S} + \mathcal{S}^\perp$. The Hamiltonian $H_2$ is such that $\mathcal{S}$ is its zero eigenspace and the eigenvectors in $\mathcal{S}^\perp$ have eigenvalue at least $J > 2\|H_1\|$. Then
    \[
        \lambda_0(H_1|_{\mathcal{S}}) - \frac{\|H_1\|^2}{J - 2\|H_1\|} \leq \lambda_0(H) \leq \lambda_0(H_1|_{\mathcal{S}})\,.
    \]
    If $H_1$ is positive semi-definite, then $\lambda_0(H) = \lambda_0(H_1|_S)$. 
\end{lemma}
\begin{proof}
Let $\ket{v} \in \cS$ be the eigenvector of $H_1|_\cS$ corresponding to the smallest eigenvalue of $H_1|_\cS$. Then
\[
    \braket{v}{H|v} = \braket{v}{H_1|v} + \braket{v}{H_2|v} = \lambda_0(H_1|_\cS)
\]
since $H_2\ket{v}=0$. This proves the upper bound.

For the lower bound, let $\ket{w} \in \mathcal{H}$ be an arbitrary vector in $\mathcal{H}$, which we can always decompose as $\ket{w} = \alpha_1\ket{w_1} + \alpha_2\ket{w_2}$, where $\ket{w_1}\in\cS$ and $\ket{w_2}\in\cS^\perp$, $\alpha_1 \geq 0,\alpha_2 \geq 0 \in \mathbb{R}$ s.t. $\alpha_1^2 + \alpha_2^2 = 1$. Let $K=\|H_1\|$. Then
\begin{eqnarray*}
    \braket{w}{H|w} &\geq& \braket{w}{H_1|w} + J\alpha_2^2 \\
    &\geq& \braket{w_1}{H_1|w_1} - K\alpha_2^2 - 2K\alpha_2 - K\alpha_2^2 + J\alpha_2^2 \\
    &\geq& \lambda_0(H_1|_\cS) + (J-2K)\alpha_2^2 - 2K\alpha_2\,,
\end{eqnarray*}
which is minimized by $\alpha_2 = J/(J-2K)$. If $H_1$ is positive semi-definite, then 
\begin{eqnarray*}
    \braket{w}{H|w} &\geq& \braket{w}{H_1|w} + J\alpha_2^2 \\
    &\geq& \braket{w_1}{H_1|w_1} - K\alpha_2^2 + J\alpha_2^2 \\
    &\geq& \lambda_0(H_1|_\cS) + (J-K)\alpha_2^2\,,
\end{eqnarray*}
which is minimized by $\alpha_2 = 0$ (since $J > K$ by assumption) and so we obtain $\braket{w}{H|w} \geq \lambda_0(H_1|_\cS)$, matching the upper bound.
\end{proof}
\noindent For instance, if $J \geq 8\|H_1\|^2 + 2\|H_1\| = \poly(\|H_1\|)$ we get $\lambda_0(H_1|_{\mathcal{S}}) - 1/8 \leq \lambda_0(H) \leq \lambda_0(H_1|_{\mathcal{S}})$.

\ 

\noindent By closely following the proof in~\cite{kempe2006}, we prove the following:
\begin{restatable}{lemma}{lemEigenvalues}\label{lem:eigenvalues}
Suppose the circuit $U$ accepts with probability $1-\epsilon$ on input $\ket{x,0}$. Then the smallest eigenvalue of $H$ is unique and satisfies $\epsilon-\frac14 \leq \lambda_0(H) \leq \epsilon$. Moreover, the state
    \begin{equation}\label{eq:history_state}
        \ket{\eta} = \frac{1}{T+M} \sum_{t=1}^{T+M} U_t\cdots U_1\ket{x,0} \otimes \ket{\hatt}
    \end{equation}
    satisfies $\braket{\eta}{ H | \eta} = \epsilon$, where $M = \poly(T)$ and $\hatt$ represents correct unary encoding of the integer $t$ as per Eq.~\eqref{eq:unary_clock}.
\end{restatable}
\noindent A full proof of Lemma~\ref{lem:eigenvalues} can be found in Appendix~\ref{app:proof_1}. We now turn our attention to lower bounding the energy of states orthogonal to $\ket{\eta}$. We obtain the following result, which in our construction plays a role analogous to that of Lemma 2 of Ref.~\cite{gharibian2021dequantizing}.
\begin{lemma} Any state orthogonal to $\ket{\eta}$ has energy at least $\Delta$, provided that  \begin{eqnarray*}
    J_{\text{clock}} &\geq& \poly(\|H_{\text{out}}+J_{\text{in}}H_{\text{in}}+J_2H_{\text{prop2}}+J_1H_{\text{prop1}}\|) \Delta  \\
    J_1 &\geq& \poly(\|H_{\text{out}}+J_{\text{in}}H_{\text{in}}+J_2H_{\text{prop2}}\|M)\Delta \\
    J_2 &\geq& \poly(\|H_{\text{out}}+J_{\text{in}}H_{\text{in}}\|T)\Delta \\
    J_{\text{in}} &\geq& \poly(\|H_{\text{out}}\|(T+M))\Delta \,.
\end{eqnarray*}
For this choice of parameters, $\norm{H} = \mO(\poly(n)\Delta)$.
\label{lem:orthogonal_states}
\end{lemma}
\begin{figure}
    \centering
    \includegraphics{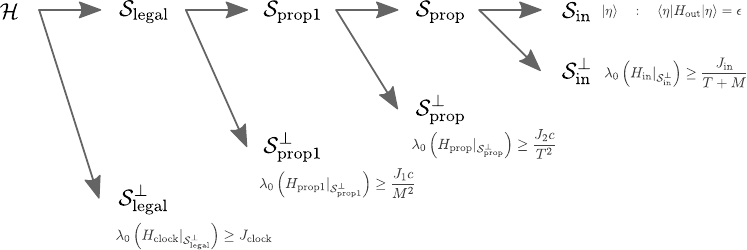}
    \caption{Summary of the proof used to lower bound $\lambda_0(H)$. At each step, we `peel off' a part of the Hilbert space until we are left with a single state spanning the groundspace of $H$. The arrows show the direction of inclusions: e.g. $\mathcal{H} \supset \cS_{\text{legal}}$ and $\mathcal{H} \supset \cS_{\text{legal}}^\perp$. Underneath the bottom row of subspaces is written the smallest eigenvalue of a particular Hamiltonian within that subspace (i.e. the minimum energy penalty given by that Hamiltonian to states contained in that subspace).}
    \label{fig:proof_summary}
\end{figure}
\begin{proof}
In Figure~\ref{fig:proof_summary} we show a summary of the steps performed in the proof of Lemma~\ref{lem:eigenvalues}: at each step, we peel off a part of the Hilbert space $\mathcal{H}$ in order to apply Lemma~\ref{lem:projection}, which lower bounds the energy given to states in a particular subspace by a particular term in the Hamiltonian. Note in particular that
\[
    \cS_{\text{legal}}^\perp \cup \cS_{\text{prop1}}^\perp \cup
    \cS_{\text{prop}}^\perp \cup \cS_{\text{in}}^\perp \cup
    \cS_{\text{in}} = \mathcal{H}\, \qquad \cS_{\text{legal}}^\perp \cap \cS_{\text{prop1}}^\perp \cap
    \cS_{\text{prop}}^\perp \cap \cS_{\text{in}}^\perp \cap
    \cS_{\text{in}} = 0\,,
\]
The smallest eigenvalue of $H$ lies in the 1-dimensional space $\cS_{\text{in}}$ (which is spanned by $\ket{\eta}$), implying that the second smallest eigenvalue can only be supported in $\cS_{\text{in}}^\perp \cup \cS_{\text{prop}}^\perp \cup \cS_{\text{prop1}}^\perp \cup \cS_{\text{legal}}^\perp$. If all terms in the Hamiltonian were positive semi-definite, then by working backwards through the subspace splitting we could conclude that the second smallest eigenvalue must have a value larger than or equal to the smallest penalty given by any of the terms $H_{\text{clock}} | _{\mathcal{S}_{\text{legal}}^\perp}$, $H_{\text{prop1}} | _{\mathcal{S}_{\text{prop1}}^\perp}$, $H_{\text{prop}} | _{\mathcal{S}_{\text{prop}}^\perp}$, or $H_{\text{in}} | _{\mathcal{S}_{\text{in}}^\perp}$. 

Unfortunately, as we note in the proof of Lemma~\ref{lem:eigenvalues}, $H_{\text{prop2}}$ is not positive semi-definite; however, within the space $\cS_{\text{prop1}}$ it is. Therefore, any state lying in $\cS_{\text{in}}^\perp \subset \cS_{\text{prop}} \subset \cS_{\text{prop1}}$ will have energy at least $\frac{J_{\text{in}}}{T+M}$. Similarly, any state in $\cS_{\text{prop}}^\perp \subset \cS_{\text{prop1}}$ will have energy at least $\frac{J_2c}{T^2}$. 

Now we turn our attention to states outside of $\cS_{\text{prop1}}$, whose energy might pick up a negative contribution from $H_{\text{prop2}}$, of magnitude at most $J_2\|H_{\text{prop2}}\|$. If the state lies in $\cS_{\text{prop1}}^\perp$, then so long as we choose $J_1 \geq J_2\|H_{\text{prop2}}\|M^2 \Delta/c$ for some $\Delta$ (to be chosen later), the state will have overall energy at least $\Delta$. Likewise, if the state lies in $\cS_{\text{legal}}^\perp$, then so long as we choose $J_{\text{clock}} \geq J_2\|H_{\text{prop2}}\|\Delta$ it will have energy at least $\Delta$. 

In line with the proof of~\ref{lem:eigenvalues}, we can choose the Hamiltonian coefficients such that
\begin{eqnarray*}
    J_{\text{clock}} &\geq& \poly(\|H_{\text{out}}+J_{\text{in}}H_{\text{in}}+J_2H_{\text{prop2}}+J_1H_{\text{prop1}}\|) \Delta  \\
    J_1 &\geq& \poly(\|H_{\text{out}}+J_{\text{in}}H_{\text{in}}+J_2H_{\text{prop2}}\|M)\Delta \\
    J_2 &\geq& \poly(\|H_{\text{out}}+J_{\text{in}}H_{\text{in}}\|T)\Delta \\
    J_{\text{in}} &\geq& \poly(\|H_{\text{out}}\|(T+M))\Delta \,.
\end{eqnarray*}
With this combination of choices, we can guarantee that any state orthogonal to $\ket{\eta}$ has energy at least $\Delta$.
\end{proof}

At this point we have a 2-local Hamiltonian $H$ and a low-energy state $\ket{\eta}$ encoding some $\BQP$ computation, as well as a lower bound on the energy of states orthogonal to $\ket{\eta}$ (which we later use to prove a bound on the spectral gap of $H$). We now show that the same argument that~\cite{gharibian2021dequantizing} make about the fidelity of the history state $\ket{\eta}$ with the true groundstate $\ket{\phi}$ in the {\sc yes}-case directly translates also to the {\sc no}-case, by noting that
\begin{align}
    \bra{\phi }H \ket{\phi} \leq \bra{\eta }H \ket{\eta} = \epsilon \leq 1\,,
    \label{eq:uboundE0}
\end{align}
and combining this with Lemma~\ref{lem:orthogonal_states}.

\begin{lemma}
\label{lem:fidelity_eta_gs}
The fidelity of the history state $\ket{\eta}$ of Eq.~\ref{eq:history_state} with the true (unknown) ground state $\ket{\phi}$ of $H$ satisfies
\begin{equation}\label{eqn:fidelity_h_g}
    |\bra{\phi}\ket{\eta}|^2 \geq 1 - \frac1\Delta
\end{equation}
\end{lemma}
\begin{proof}
Since $\braket{\eta}{H|\eta} = \epsilon \leq 1$ in both the {\sc yes} and {\sc no} cases, we can prove the result for both simultaneously. Let
\begin{align}
\begin{split}
     &H_1 := J_{\text{in}} H_{\text{in}} + J_{\text{clock}} H_{\text{clock}} + J_1 H_{\text{prop1}} + J_2 H_{\text{prop2}},\\
     &H_2 := (T+M) H_\text{out} 
     \end{split}
     \label{eq:H1_H2_2local}
\end{align}
(i.e. everything except the term $H_{\text{out}}$). Write $\ket{\phi} = \alpha_1 \ket{\eta} + \alpha_2 \ket{\eta^\bot}$, for $\ket{\eta}$ the history state from Eq.~\ref{eq:history_state}, and $|\alpha_1|^2 + |\alpha_2|^2 = 1$. Then,
\begin{align*}
    1 &\geq \bra{\phi}H\ket{\phi} &&\\
    &\geq \bra{\phi} H_1\ket{\phi} \quad &&(H_2 \succeq 0)\\
    &= ( |\alpha_1|^2 \bra{\eta} H_1\ket{\eta} + |\alpha_2|^2 \bra{\eta^\perp} H_1\ket{\eta^\perp}) &&\\
    &= |\alpha_2|^2 \bra{\eta^\perp} H_1\ket{\eta^\perp} \quad && (H_1 \ket{\eta} = 0)\\
    &\geq |\alpha_2|^2 \Delta\, \quad  &&\text{(Lemma~\ref{lem:orthogonal_states})}
\end{align*}
where we use the fact that $H_1\ket{\eta}=0$. From this we conclude 
\begin{equation}\label{eqn:fidelity_h_phi_nocase}
    |\bra{\phi}\ket{\eta}|^2\geq 1 - \frac{1}{\Delta}.
\end{equation}
\end{proof}
We note here that we can prove a lower bound on the spectral gap of $H$ in a similar manner, by using Lemma~\ref{lem:orthogonal_states}.
\begin{corollary} The spectral gap $\gamma(H)$ is lower bounded by
\begin{align}
    \gamma (H) := \lambda_1(H)-\lambda_0(H) \geq \Delta-2
    \end{align}
\label{lem:specgapH}
\end{corollary}
\begin{proof}
The proof goes in similar fashion to the one used to prove Lemma~\ref{lem:fidelity_eta_gs}. Let $\{\phi,\phi_1,\dots,\phi_{2^n-1} \}$ be the set of eigenstates of $H$, ordered from lowest to highest energy.  We can write the history state in the eigenbasis of $H$ as 
\[
\ket{\eta} = a \ket{\phi} + b \ket{\phi_1} + c \ket{\phi_\text{rest}},
\]
where $\ket{\phi_\text{rest}} \in \text{span}\{\phi_2,\dots,\phi_{2^n-1}\}$ and $\abs{a}^2 + \abs{b}^2 + \abs{c}^2 = 1$. Since $\abs{c}^2 \geq 0$, and we must have that $\abs{\bra{\eta}\ket{\phi_1}}^2= \abs{b}^2 \leq 1-\abs{a}^2  = 1- \abs{\bra{\eta}\ket{\phi}}^2$, for which a lower bound on $\abs{\bra{\eta}\ket{\phi}}^2$ is known from Lemma~\ref{lem:fidelity_eta_gs}.  We rewrite $\ket{\phi_1}$ as
\[
\ket{\phi_1} = b \ket{\eta} + \beta \ket{\eta^\perp},
\]
where $\abs{\beta}^2 = 1-\abs{b}^2 \geq \abs{\bra{\eta}\ket{\phi}}^2$. Again, let $H_1$ and $H_2$ be as in Eq.~\eqref{eq:H1_H2_2local}. A lower bound on the first excited state energy $\lambda_1$ is then 
\begin{align*}
\lambda_1(H)&=\bra{\phi_1} H \ket{\phi_1} \\
&\geq \bra{\phi_1} H_1 \ket{\phi_1} \quad &&\text{(Since $H_2 \succeq 0$)}\\
&= \abs{b}^2 \bra{\eta} H_1 \ket{\eta} + \abs{\beta}^2 \bra{\eta^\perp} H_1 \ket{\eta^\perp}  &&\text{($\ket{\eta},\ket{\eta^\perp}$ are eigenstates of $H_1$)}\\
&=\abs{\beta}^2 \bra{\eta^\perp} H_1 \ket{\eta^\perp} &&\text{($H_1 \ket{\eta}=0)$}\\
&\geq \abs{\bra{\eta}\ket{\phi}}^2 \Delta \quad &&\text{(Lemma~\ref{lem:orthogonal_states})}\\
&\geq \Delta-1  \quad &&\text{(Lemma~\ref{lem:fidelity_eta_gs}).}
\end{align*}
If we combine the above with the already established upper bounds on $\lambda_0(H)$ of Eq.~\eqref{eq:uboundE0}, we arrive at the desired result.
\qedhere
\end{proof}
Returning to the result of Lemma~\ref{lem:fidelity_eta_gs}, note that the fidelity $|\braket{\phi}{\eta}|^2 \geq 1-\frac1\Delta$ can be increased to $1-\frac{1}{\poly(n)}$ by setting $\Delta = poly(n)$, which still keeps $\|H\|=\mO(\poly(n))$. Finally, the fidelity of $\ket{\eta}$ with the guiding state $\ket{u}$ is
\begin{align}
    \abs{\braket{\eta}{u}}^2 &= \abs{\sum_{t=1}^M \frac{1}{\sqrt{M(M+T)}} }^2 \nonumber\\ 
    &= \frac{M^2}{M^2+MT} \nonumber \\
    &\geq 1-\frac{T}{M}, \label{eq:fidelity_eta_u}
\end{align}
by taking the series expansion at $1/M \rightarrow 0 $ and assuming $M \gg T$. And so we have 
\begin{align}
    ||\proj{u} - \proj{\phi}|| &\leq ||\proj{u} - \proj{\eta}|| + ||\proj{\eta} - \proj{\phi}|| \nonumber\\
    2\sqrt{1 - |\bra{u}\ket{\phi}|^2} &\leq 2\sqrt{1 - |\bra{u}\ket{\eta}|^2} + 2\sqrt{1 - |\bra{\eta}\ket{\phi}|^2} \nonumber\\
    |\bra{u}\ket{\phi}|^2  &\geq  1 - \left(\sqrt{1 - |\bra{u}\ket{\eta}|^2} + \sqrt{1 - |\bra{\eta}\ket{\phi}|^2}\right)^2 \nonumber\\
    &\geq 1 - \left(\sqrt{\frac{T}{M}} + \sqrt{\frac{1}{\Delta}}\right)^2\,. \label{eq:fidelity_u_gs}
\end{align}
By setting $\Delta = \poly(n)$ and $M=\poly(T)$ (where $T=\mO(\poly(n))$ we have that the guiding vector always has an fidelity of at least $1 - 1/(poly(n))$ with the true groundstate, and circumvents the need to split the Hilbert space and therefore for the locality to be increased, as was the case in the original construction of~\cite{gharibian2021dequantizing}, and hence we obtain $\BQP$-hardness for a 2-local Hamiltonian. This is sufficient to prove Theorem~\ref{thm:main_result1} for the case $c=0$. In the next section we prove the results for the case $0<c<\mO(\poly(n))$.

%% file: excited.tex
\section{Generalization to excited state energies}\label{sec:excited}
In Ref.~\cite{PhysRevA.81.032331}, the authors show that determining the $c$th excited state energy of a $k$-local Hamiltonian ($k\geq 3$), where $c = \poly(n)$, is $\QMA$-complete -- even if all the $c-1$ energy eigenstates and corresponding energies are known. We will now show that their construction easily translates to the setting with guiding states. This follows rather straightforwardly from combining the results of the two previous sections with established results in Hamiltonian complexity theory, in particular the aforementioned $3$-local $\QMA$-complete excited state Hamiltonian from Ref.~\cite{PhysRevA.81.032331} and perturbative gadget techniques from Ref.~\cite{kempe2006}. As a bonus, this also shows that the unguided problem is $\QMA$-hard for $k=2$, which was left open in~\cite{PhysRevA.81.032331}.

\subsection{A 3-local gadget Hamiltonian for low-energy states}
We first prove $\BQP$-hardness of $\mathsf{GLHLE}(k,c,\epsilon,\zeta)$ with $k\geq 3$, which follows simply from combining Eq.~\eqref{eq:fidelity_u_gs} with a construction similar to the one used in Ref.~\cite{PhysRevA.81.032331}.
\begin{lemma} 
$\mathsf{GLHLE}(k,c,\epsilon,\zeta)$ is $\BQP$-hard for $k\geq 3$, $c \leq \mO(\poly(n))$ and $\zeta = 1-\Omega(1/\poly(n))$. 
\end{lemma}
\begin{proof}
We will reduce directly from the $\BQP$-complete Hamiltonian $H$ as defined in Eq.~\eqref{eq:H_2local}. Again, let $\ket{u}$ be a semi-classical guiding state such that $\abs{\bra{u} \ket{\psi_0}} \geq \zeta$. Consider the following $3-$local Hamiltonian $H^{(c)}$ on $n+1$ qubits\footnote{Note that this gadget can be trivially changed such that estimating the $n$ highest energy states is $\BQP$-hard.}:
\begin{align}
    H^{(c)} = H^{(z)} \otimes \ket{0} \bra{0} +  H^{(s)} \otimes \ket{1}\bra{1},
    \label{eq:Hc}
\end{align}
where 
\begin{align*}
    &H^{(z)} = \sum_{i=0}^{d} 2^{i} \kb{1}{1}_i + \sum_{i=d+1}^n 2^{d+1}  \kb{1}{1}_i  -\left( c-\frac{1}{2}\right)I,\\
    &H^{(s)} =  \frac{1}{2} \frac{H+I/4}{\norm{H}+1/4} -\frac{1}{4} I,\\
\end{align*}
with $d= \lceil \log2(c) \rceil$. Note that $H^{(z)}$ has exactly $c$ states with negative energy, with the smallest eigenvalue being $-c+\frac{1}{2}$ and the largest eigenvalue valued at $\sum_{i=0}^d 2^i + \sum_{i=d+1}^{n} 2^{d+1} - \left(c-\frac{1}{2} \right) =  2^{d+1}  + 2^{d+1} (n - d) -\frac{1}{2}-c$. The spectrum jumps in integer steps of $1$, and has as largest negative (resp. smallest non-negative) energy value $-\frac{1}{2}$ (resp. $\frac{1}{2}$). Since $\text{eig}(H^{(s)})\in[-1/4,1/4]$, we must have that $H^{(s)}$ sits precisely at the $c$th excited state level (or $c+1$th eigenstate level) in $H^{(c)}$. Therefore, given a guiding state $\ket{u}$ for $H$ such that $\abs{\bra{u}\ket{\psi_0}} \geq \zeta$, one has that the guiding state $\lvert u^{(c)} \rangle  = \ket{u} \otimes  \ket{1}$ is also semi-classical and must have $| \langle u^{(c)} \lvert \psi_c^{(c)} \rangle| \geq \zeta$, where $\lvert \psi_c^{(c)} \rangle$ denotes the $c$th excited state of $H^{(c)}$. Since this construction of $H^{(c)}$ and $\lvert u^{(c)} \rangle $ provides a polynomial time reduction from  an instance of $\mathsf{GLHLE}(k,0,\epsilon,\zeta)$  to one of $\mathsf{GLHLE}(k,c,\epsilon,\zeta)$, whenever $c=\mO(\poly(n))$, we must have that $\mathsf{GLHLE}(k,c,\epsilon,\zeta)$  is $\BQP$-hard whenever $k\geq 3$. \qedhere
\end{proof}
What remains to be done is to bring the locality down from $k=3$ to $k=2$, which will be shown in the next section.

\subsection{Reducing the locality}
We will use the $3$-to-$2$-local perturbative gadget, introduced in the same paper as the 2-local construction we used before~\cite{kempe2006}.  The gadget construction starts with the following lemma:
\begin{lemma}[Lemma 8 from~\cite{kempe2006}] Any $3$-local Hamiltonian $H^{(3)}$ on $n$ qubits can be re-written as 
\begin{align*}
    H^{(3)} = c_r \left( Y-6 \sum_{m=1}^M B_{m1} B_{m2} B_{m3} \right),
\end{align*}
where $Y$ is a 2-local Hamiltonian with $\norm{Y} = \mO(1/n^6)$, $M=\mO(n^3)$, each $B_{mi}$ is a one-qubit term of norm $\mO(1/n^3)$ that satisfies $B_{mi} \geq I/n^3$, and $c_r$ is a rescaling factor that satisfies $1\leq c_r \leq \poly(n)$.
\label{lem:H3}
\end{lemma}
Now, define another Hamiltonian $H_\text{eff}$ on $n+3M$ qubits in Hilbert space  $\mathcal{H}_\text{eff} = (\mathbb{C}^{2})^{\otimes n} \otimes (\mathbb{C}^{2})^{\otimes 3M}$ as
\begin{align*}
    H_\text{eff}:= c_r\left( Y \otimes I_C-6 \sum_{m=1}^M B_{m1} B_{m2} B_{m3}  \otimes (\sigma^x_m)_C \right),
\end{align*}
where $C := \text{Span} (\ket{000},\ket{111})$ represents a logical qubit on which Pauli operators $\sigma^i_c$ can act. Let $\mathcal{H}_\text{eff} = \mathcal{P}_\text{+}\oplus\mathcal{P}_\text{else}$, where $\mathcal{P}_\text{+}$ is the subspace spanned by all eigenvectors of $H_\text{eff}$ in which all ancillary logical qubits are in the $\ket{+}$ state and $\mathcal{P}_\text{else}$ spanned by all other eigenvectors. Observe that we have that
\begin{align}
     \text{eig} \left( H_\text{eff}|_{\mathcal{P}_\text{+}} \right) = \text{eig} \left(H^{(3)} \right).
     \label{eq:low_energy_equiv}
\end{align}
The key idea is now that the low-energy sector of $ H_\text{eff}$ -- which is precisely the one corresponding to the eigenvalues that are identical to those of $H^{(3)}$ -- can be approximated by a 2-local gadget Hamiltonian 
\begin{align}
    \tilde{H} = Q + P,
    \label{eq:3to2gadget}
\end{align}
where 
\begin{align*}
    &Q = -\frac{c_r}{4 \mu^3} \sum_{m=1}^M I \otimes (\sigma^z_{m1} \sigma^z_{m2} + \sigma^z_{m1} \sigma^z_{m3} + \sigma^z_{m2} \sigma^z_{m3} - 3I),\\
    &P = c_r\left[Y+ \sum_{m=1}^M  \left(\frac{1}{\mu}(B_{m1}^2+B_{m2}^2+B_{M3}^2) \otimes I -\frac{1}{\mu^2} (B_{m1} \otimes \sigma_1^x + B_{m2} \otimes \sigma_2^x + B_{m3} \otimes \sigma_3^x)\right)\right],
\end{align*}
where $\mu > 0$ is some sufficiently small constant. Note that $Q$ has eigenvalues $0$ and $c_r \mu^{-3}$, and therefore a spectral gap of $\Lambda := c_r \mu^{-3}$. Since $\norm{B_{mi}} \leq \mO(1/n^3)$ and $M = \mO(n^3)$ we have that $\norm{P} = \mO(c_r \mu^{-2}) \leq \Lambda/2$. By perturbative analysis of the self-energy\footnote{See Appendix~\ref{app:proofs_2} for the definition, which is not relevant to understand the results in the main text.} of $\tilde{H}$, one can relate the spectra of $H_\text{eff}$ and $\tilde{H}$ in the following way:

\begin{lemma} [Modified Theorem 3 from Ref.~\cite{kempe2006}.]\label{lem:delta_close} Let $\Lambda:=c_r \mu^{-3}$ and $\tilde{H} = Q+P$ as in Eq.~\ref{eq:3to2gadget} with $\norm{P} = \mO(c_r \mu^{-2}) \leq \Lambda/2$. Let $\lambda_{*} = \Lambda/2$, $\lambda_+ = \lambda_{*} + \Lambda/2$ and $\lambda_- = \lambda_{*} - \Lambda/2$.  Then we have that 
\begin{align}
    | \lambda_j (\tilde{H}|_{\tilde{\mathcal{L}_-}}) - \lambda_j (H_\text{eff}) | \leq c_r \mu, 
\end{align}
for all $j$.
\end{lemma}

\begin{figure}
    \centering
    \includegraphics{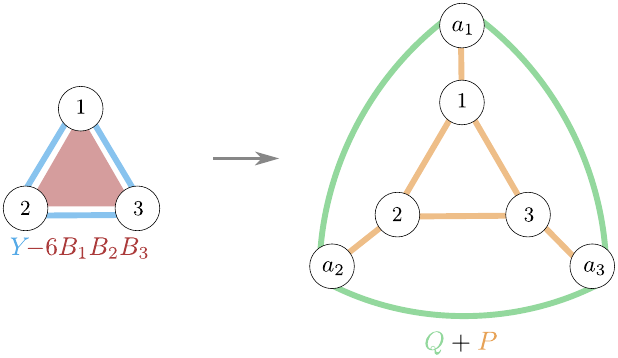}
    \caption{Visualization of the perturbative gadget for a Hamiltonian consisting of a single $3$-local term. The ancilla registers are indicated with `$a_i$'. }
    \label{fig:gadgetvisual}
\end{figure}

\subsubsection{Application to the low-energy gadget}
Define the minimum relative spectral gap $\gamma_i(H)$ of eigenstate $i$ of Hamiltonian $H$ as
\begin{align}
    \gamma_i(H) = \min_{j\neq i} \abs{\lambda_j(H) - \lambda_i(H)}.
\end{align}
In order to apply the $3$-to-$2$-local construction, we first need to establish the following:
\begin{enumerate}
    \item A lower bound on the fidelity between \emph{any eigenstate} of $H_\text{eff}$ and $\tilde{H}$ provided that their energies in the low-energy sector are $c_r \mu$-close, and not just between the groundstates (which is what is given in Ref.~\cite{kempe2006}).
    \item Since this turns out to depend on the relative energy differences between adjacent energy levels, it is sufficient to know all $\gamma_{i} (H_\text{eff})$ for eigenstates up to $i=c$.
\end{enumerate}
Let us now address both points. For the latter, observe that we actually already know all the relative energy gaps of $H_\text{eff}$ up to, but not including, $c$ from the way $H^{(c)}$ is constructed. We also then have a bound on the spectral gap of $H$ from Corollary~\ref{lem:specgapH}, $\gamma(H) \geq \Delta-2 = \poly(n)$, which gives the remaining relative energy gap.

Now we consider the first point, which is addressed by the following lemma.
\begin{restatable}{lemma}{lemmaNine}
Assume that $H_\text{eff}$, $Q$, $P$ satisfy the conditions of Lemma~\ref{lem:delta_close} with some $\mu > 0$. Let $\ket{\tilde{v}_i}$ (resp., $\ket{v_{\text{eff},i}}$) denote the $i$th eigenstate of $\tilde{H}$ (resp., $H_\text{eff}$). Let $\gamma_i (H_\text{eff})$ denote the minimum relative spectral gap of $H_\text{eff}$. Then we have that
\begin{align}
    \abs{\bra{\tilde{v}_i}\ket{v_{\text{eff},i}}}^2 &\geq 1-\left(\frac{\norm{P}}{c_r \mu^{-3}-\lambda_i(H_\text{eff})-c_r\mu} + \sqrt{\frac{2 c_r\mu}{\gamma_i (H_\text{eff})}} \right)^2,
\end{align}
for all $i$.
\label{lem:fidelity_pert}
\end{restatable}
\noindent The proof of Lemma~\ref{lem:fidelity_pert} is given in Appendix~\ref{app:proofs_2}. We now have all the tools at our disposal to prove the following theorem.
\begin{theorem} 
$\glh$ is $\BQP$-hard for $k\geq 2$, $ 1 \leq c \leq \mO(\poly(n))$, $\delta = 1/\poly(n)$ and $\zeta = 1-1/\poly(n)$. 
\label{thm:excited_states}
\end{theorem}
\begin{proof}
Let $H^{(c)}$ be the Hamiltonian as in Eq.~\eqref{eq:Hc}. We first construct $H^{(3)}$ from $H^{(c)}$ according to Lemma~\ref{lem:H3}, and next $H_\text{eff}$ from $H^{(3)}$. From this we then construct the 2-local Hamiltonian $\tilde{H}$. The steps of the proof are summarized in Figure~\ref{fig:reductionvisual}.

We have that for $H^{(3)}$ (and therefore for $H_\text{eff}$) that
\begin{align}
    \gamma_i \left(H^{(3)}\right) \geq
    \min\left[\frac{\gamma(H)}{\norm{H}+1/4} ,1/4\right] \text{ for } i \leq c,
\end{align}
since the spectrum of $H^{(3)}$ is such that it jumps in increments of 1 until the jump to the groundstate of $H$, which is at least $1/4$, followed by the next jump to the first excited state of $H$, which is at least $\frac{\gamma(H)}{\norm{H}+1/4}$. A bound on $\gamma(H)$ is known from Lemma~\ref{lem:specgapH}. By Lemma~\ref{lem:delta_close} we have that
\begin{align*}
    | \lambda_j (\tilde{H}|_{\tilde{\mathcal{L}_-}}) - \lambda_j (H_\text{eff}) | \leq c_r \mu,
\end{align*}
provided that $\mu$ is chosen such that all conditions of Lemma~\ref{lem:delta_close} hold. Additionally, $\mu$ should also be small enough such that the ordering of the eigenvalues of $\tilde{H}$ in the low energy sector respects the original ordering of those in $H_\text{eff}$.  Setting $\mu = 1/\poly(n)$ such that
\begin{align*}
   \mu < \frac{1}{2} \frac{\gamma(H)}{\norm{H}+1/4},
\end{align*}
is sufficient. Let $\ket{\tilde{u}_{c}}=\ket{u} \otimes \ket{1} \otimes \ket{+}^M$, which is a semi-classical state since $\ket{u}$ is semi-classical. Then, by  Lemma~\ref{lem:fidelity_pert}, we have that
\begin{align*}
    \abs{\bra{\tilde{v}_c}\ket{v_{\text{eff},c}}}^2 &\geq 1-\left(\frac{\norm{P}}{c_r \mu^{-3}-\lambda_c(H_\text{eff})-c_r\mu} + \sqrt{\frac{2 c_r\mu}{\gamma_c (H_\text{eff})}} \right)^2,
\end{align*}
which can be made $\geq 1-1/\poly(n)$ since $\norm{P} = \mO(c_r \mu^{-2})$, $\lambda_c(H_\text{eff}) = \mO(1)$. This in combination with Eq.~\eqref{eq:fidelity_u_gs} with $\ket{\phi}=\ket{v_{\text{eff},c}}$ and $\ket{u}=\ket{\tilde{u}_c}$ gives us
\begin{align*}
    \abs{\bra{\tilde{v}_c}\ket{\tilde{u}_c}}^2 &\geq 1-\left(\sqrt{1-\abs{\bra{\tilde{v}_c}\ket{v_{\text{eff},c}}}^2}  + \sqrt{1-\abs{\bra{\tilde{u}_c }\ket{v_{\text{eff},c}}}^2}\right)^2,
\end{align*}
which again can be made $\geq 1-1/\poly(n)$. Since fidelity is invariant under changes to the norm of $\tilde{H}$, we can scale down $\tilde{H}$ such that its operator norm becomes smaller than 1. We have that $\norm{\tilde{H}} = \norm{Q+P} \leq \norm{Q} + \norm{P} = \mO(c_r \mu^{-3}) = \mO(\poly(n))$, so the re-scaling factor need only be inverse polynomial in $n$.
\qedhere
\end{proof}

\begin{figure}
    \centering
    \includegraphics[width=\linewidth]{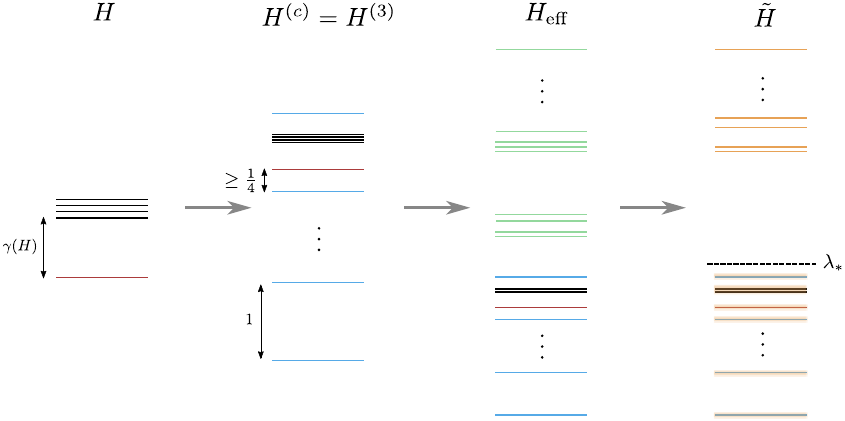}
    \caption{Visualization of all steps in the reduction used in the proof of theorem~\ref{thm:excited_states}.}
    \label{fig:reductionvisual}
\end{figure}

\noindent This concludes the proof of Theorem~\ref{thm:main_result}, which follows straightforwardly from Lemma~\ref{lem:eigenvalues} and Theorem~\ref{thm:excited_states}.

%% file: containment.tex
\section{Containment in $\BQP$}\label{sec:containment}
In this section we show that $\glh$ is contained in $\BQP$ -- i.e. we prove Theorem~\ref{thm:main_result2}, which we restate below for convenience. 
First, we recall some basic facts about combining Hamiltonian simulation with quantum phase estimation.
\begin{lemma}[Quantum eigenvalue estimation]\label{thm:phase_estimation}
Let $H$ be an $\mO(\log n)$-local Hamiltonian acting on $n$ qubits, with eigenvectors $\ket{\psi_j}$ and corresponding eigenvalues $\lambda_j \in [0,1]$. Then there is a quantum algorithm that, given as input an eigenvector $\ket{\psi_j}$, will output with probability at least $p$ an $\epsilon$-approximation of $\lambda_j$ (i.e. an estimate $\tilde{\lambda_j}$ such that $|\tilde{\lambda_j} - \lambda_j| \leq \epsilon$) in time $\poly(n,1/\epsilon,1/p)$.
\end{lemma}
\noindent This is by now a commonly used quantum algorithm; for details and proofs of correctness, see e.g.~\cite{gharibian2021dequantizing,cade2018quantum}.
\bqpcontainment*
\begin{proof}
Recall that the fidelity of the guiding state with the target eigenstate is at least $\zeta$. Containment in $\BQP$ follows from the standard quantum algorithm of Lemma~\ref{thm:phase_estimation}. 
If we input an arbitrary $n$-qubit state $\ket{\phi}$ to the algorithm of Lemma~\ref{thm:phase_estimation}, it follows that we will obtain an $\epsilon$-approximation of $\lambda_j$ with probability $\geq p|\braket{\psi_j}{\phi}|^2$, and hence if we input the guiding state $\ket{u}$, we will obtain an $\epsilon$-approximation to the target eigenstate with probability $\geq p\zeta$. For the case $c=0$, with $\zeta = \Omega(1/\poly(n))$, we can therefore obtain an $\epsilon$-approximation to the groundstate energy with probability $\geq 1/r(n)$, for $r$ some polynomial, in time $\poly(n,1/\epsilon)$. To distinguish the case that the groundstate energy is $\leq a$ or $\geq b$, with $b-a\geq \delta$, setting $\epsilon < \delta$ is sufficient. With $\delta = 1/\mO(\poly(n))$, this takes time $\poly(n)$, proving part (i) of the theorem. 

For the case $c > 0$, and with $\zeta = \frac12 + \Omega(1/\poly(n))$, we can choose $p > 1-1/\poly(n)$ sufficiently large so that with probability at least $p\zeta > 1 - 1/\poly(n)$ we obtain an $\epsilon$-approximation to the target energy $\lambda$. Again by choosing $\epsilon < \delta$ we can decide whether $\lambda \leq a$ or $\lambda \geq b$ with probability $\geq \frac12 + 1/\poly(n)$. By repeating $\poly(n)$ times and taking a majority vote, we can decide which is the case with probability, say, $2/3$ by a Chernoff bound, proving part (ii) of the theorem.

\end{proof}

%% file: conclusion.tex
\section{Discussion, conclusion and outlook}
In this work, we have generalized and strengthened the results from~\cite{gharibian2021dequantizing}. We introduced the $\mathsf{GLHLE}$-problem, which generalizes their $\mathsf{GLH}$ problem to include excited states, and improved the $\BQP$-completeness result so that it holds over a larger range of parameter settings, in particular: for the locality (from $k\geq 6$ to $k\geq 2$), promise on the fidelity of the guiding state  with the ground state (upper bound from $1/2-\Omega(1/\poly(n)$ to $1-\Omega(1/\poly(n)$) and the considered eigenstate (from just $c=0$ to any $c=\mO(\poly(n))$).

All constructions used throughout this work rely on the Feynman-Kitaev circuit-to-Hamiltonian mapping, and therefore do not exhibit any particular structure found in physical systems except from locality constraints. As future work, it would be interesting to see if Hamiltonians more closely related to actual physical systems, such as the electronic structure Hamiltonian, still adhere to this $\BQP$-Hardness result. Recent work~\cite{o2021electronic} has indeed shown that the canonical $\QMA$-hardness results for the local Hamiltonian problem do indeed hold for electronic structure Hamiltonians, and so we expect that the $\BQP$-hardness results would hold also.

\paragraph{Where does the $\BQP$-hardness lie?}
In Section 1.4 of~\cite{gharibian2021dequantizing} the authors argue that the $\BQP$-hardness of the problem lies in the fact that the required precision is inverse-polynomial, since their dequantized classical algorithm is efficient when the required precision is merely a constant. However, recent work by Stroeks, Helsen and Terhal \cite{stroeks2022spectral} adds some more detail to the picture: there, the authors show that, given a local Hamiltonian with $1/\poly(n)$ spectral gap and a semi-classical guiding state that has non-negligible ($\Omega(1/\poly(n))$) fidelity with the ground state and \textit{at most $\mO(1)$ other excited states}, the ground state energy can be computed up to inverse polynomial precision in polynomial time.\footnote{The authors make use of a slightly different access model than in~\cite{gharibian2021dequantizing}.} This result suggests that the $\BQP$-completeness lies not only in the required precision, but \textit{also} in the properties of the guiding state -- namely that it must fidelity significantly with the ground-state, but also with many (i.e. $\omega(1)$) other eigenstates.

Our result shows that the $\mathsf{GLHLE}$-problem is $\BQP$-hard given a semi-classical guiding state even when it has as large as $1-1/\poly(n)$ fidelity with the ground state. Therefore, it must be that in the remaining $1/\poly(n)$ amplitude, the guiding state has a fidelity with $\Omega(\poly(n))$ many other eigenstates -- otherwise the techniques from~\cite{stroeks2022spectral} would be sufficient to solve the problem classically, and we would have $\BQP \subseteq \mathsf{BPP}$. To confirm this, one would need to study the form of the excited states of the Kitaev clock-Hamiltonian, for which we do not currently have a very clear understanding.



%% file: appendix.tex
\appendix 
\section{Technical proofs}
\label{app:proofs}

\subsection{Proof of Lemma~\ref{lem:eigenvalues}}\label{app:proof_1}
\lemEigenvalues*
\begin{proof}
If the circuit $V$ accepts input $\ket{x,0}$ with probability $1-\epsilon$, then the (standard) history state
    \[
        \ket{\eta} = \frac{1}{\sqrt{T+M}} \sum_{t=1}^{T+M} U_t\cdots U_1\ket{x,0}\otimes\ket{\hat{t}}
    \]
satisfies $\braket{\eta}{ H | \eta} = \epsilon$, implying that the groundstate of $H$ has energy $\leq \epsilon$, proving the upper bound. This can be checked by direct calculation, by verifying that
\[
    \braket{\eta}{H_{\text{clock}}|\eta} = \braket{\eta}{H_{\text{prop1}}|\eta} = \braket{\eta}{H_{\text{prop2}}|\eta} = \braket{\eta}{H_{\text{in}}|\eta} = 0\,
\]
and that $\braket{\eta}{H_{\text{out}}|\eta} = \epsilon$. The fact that $\braket{\eta}{H_{\text{prop2}}|\eta}$=0 is easier to see later in the proof, when we consider the restriction of $H_{\text{prop2}}$ to the nullspace $S_{\text{prop1}}$ of $H_{\text{prop1}}$. However, it is worth pointing out already that $H_{\text{prop2}}$ is \emph{not} positive semi-definite, meaning that $\braket{\eta}{H_{\text{prop2}}|\eta} = 0$ does not imply that $H_{\text{prop2}}\ket{\eta}=0$ (which does not in fact hold). In~\cite{kempe2006}, $H_{\text{prop2}}$ is eventually `replaced' by an effective (3-local) Hamiltonian that \emph{is} positive semi-definite within a particular relevant subspace. 

In the remainder of the proof we prove the lower bound of $\epsilon - \frac12$, by following closely the proof given in~\cite{kempe2006}. The approach taken there proceeds by repeatedly applying the projection lemma (Lemma~\ref{lem:projection}) to cut out pieces of the Hilbert space, in the following order
\[
    \mathcal{H} \supset \cS_{\text{legal}} \supset \cS_{\text{prop1}} \supset \cS_{\text{prop}} \supset \cS_{\text{in}}\,
\]
where $\cS_{\text{legal}}$ corresponds to the space spanned by legal clock states of the form $\ket{\hat{t}}$ given in Eq.~\eqref{eq:unary_clock}; $\cS_{\text{prop1}}$ to the space spanned by states in $\cS_{\text{legal}}$ that encode correction propagation according to all single-qubit gates (i.e. at time steps in $T_1$); $\cS_{\text{prop}}$ is the space spanned by states in $\cS_{\text{prop1}}$ that encode correct propagation via two-qubit gates (i.e. correct propagation of the regularly applied $CZ$ gates); and finally $\cS_{\text{in}}$ is the space spanned by the nullspace of $H_{\text{in}}|_{\cS_{\text{prop}}}$.

\paragraph{Restriction to $\cS_{\text{legal}}$} We begin by restricting the full Hamiltonian to the $2^n (T+M)$-dimensional space $\cS_{\text{legal}}$ spanned by states with a valid clock component of the form $\ket{\hat{t}}$. We then apply Lemma~\ref{lem:projection} with
\[
    H_1 = H_{\text{out}} + J_{\text{in}}H_{\text{in}} + J_2H_{\text{prop2}} + J_1H_{\text{prop1}} \qquad H_2 = J_{\text{clock}}H_{\text{clock}}\,.
\]
$\cS_{\text{legal}}$ is precisely the 0-eigenspace of $H_2$, and any state orthogonal to $\cS_{\text{legal}}$ has energy at least $J_{\text{clock}}$. Choosing $J_{\text{clock}} = \poly(\|H_1\|) = \poly(n)$ is sufficient to lower bound $\lambda_0(H)$ by $\lambda_0(H_1|_{\cS_{\text{legal}}}) - \frac18$. We can now concern ourselves only with lower bounding $\lambda_0(H_1|_{\cS_{\text{legal}}})$. The Hamiltonian $H_1|_{\cS_{\text{legal}}}$ takes the form
\[
  H_{\text{out}}|_{\cS_{\text{legal}}} + J_{\text{in}}H_{\text{in}}|_{\cS_{\text{legal}}} + J_2H_{\text{prop2}}|_{\cS_{\text{legal}}} + J_1H_{\text{prop1}}|_{\cS_{\text{legal}}}\,,  
\]
with
\begin{align*}
    H_{\text{in}}|_{\cS_{\text{legal}}} &= \sum_{i=0}^{m-1} \proj{x_i}_i \otimes \proj{\hat{0}} \\
    H_{\text{out}}|_{\cS_{\text{legal}}} &= (T+M)\proj{0}_0 \otimes \proj{\hat{T+M}} \\
    H_{\text{prop,t}}|_{\cS_{\text{legal}}} &= \frac12 \left( I \otimes \proj{\hatt} + I \otimes \proj{\hatt{-1}} - U_t \otimes\ketbra{\hatt}{\hatt{-1}} - U_t^\dagger \otimes \ketbra{\hatt{-1}}{\hatt} \right) \\
    H_{\text{qubit,t}}|_{\cS_{\text{legal}}} &= \frac12 \left( -2\proj{0}_{f_t} - 2\proj{0}_{s_t} + \proj{1}_{s_t} + \proj{1}_{s_t} \right) \otimes \left( \ketbra{\hatt}{\hatt{-1}} + \ketbra{\hatt{-1}}{\hatt} \right) \\
    H_{\text{time,t}}|_{\cS_{\text{legal}}} &= \frac18 I \otimes (\proj{\hatt} + 6\proj{\hatt{+1}} + \proj{\hatt{+2}} \\
    &\hspace{13mm}+ 2\ketbra{\hatt{+2}}{\hatt} + 2\ketbra{\hatt}{\hatt{+2}} + \ketbra{\hatt{+1}}{\hatt} + \ketbra{\hatt}{\hatt{+1}} + \ketbra{\hatt{+2}}{\hatt{+1}} + \ketbra{\hatt{+1}}{\hatt{+2}} \\
    &\hspace{13mm}+ \proj{\hatt{-3}} + 6\proj{\hatt{-2}} + \proj{\hatt{-1}} \\
    &\hspace{13mm}+ 2\ketbra{\hatt{-1}}{\hatt{-3}} + 2\ketbra{\hatt{-3}}{\hatt{-1}} + \ketbra{\hatt{-2}}{\hatt{-3}} + \ketbra{\hatt{-3}}{\hatt{-2}} + \ketbra{\hatt{-1}}{\hatt{-2}} + \ketbra{\hatt{-2}}{\hatt{-1}} )\,.
\end{align*}

\paragraph{Restriction to $\cS_{\text{prop1}}$} The next step is another application of Lemma~\ref{lem:projection}, now with
\[
    H_1 = (H_{\text{out}} + J_{\text{in}}H_{\text{in}} + J_2H_{\text{prop2}})|_{\cS_{\text{legal}}} \qquad H_2 = J_1 H_{\text{prop1}}|{\cS_{\text{legal}}}\,.
\]
Define the family of history states
\begin{equation}\label{eq:history_l}
    \ket{\eta_{l,i}} := \frac{1}{\sqrt{L}} \sum_{t=lL+M}^{(l+1)L+M-1} U_t\dots U_1\ket{i} \otimes \ket{\hatt}\,
\end{equation}
and 
\begin{equation}\label{eq:history_idle}
    \ket{\eta_{\text{idle},i}} := \frac{1}{\sqrt{M}} \sum_{t=1}^{M} \ket{i} \otimes \ket{\hatt}\,,
\end{equation}
where $i \in \{0,\dots,2^n-1\}$ identifies an arbitrary state in the computational basis, and in Eq.~\ref{eq:history_l} $l \in \{0,\dots,T_2\}$ picks out the block of time-steps during which single qubit gates are applied between the $l$th and $(l+1)$th 2-qubit gates. The states of Eq.~\ref{eq:history_idle} represent correct propagation during the $M$-step idling period before any gates are applied. 

Then the states from Eqs.~\ref{eq:history_l}~and~\ref{eq:history_idle} are 0-eigenvectors of $H_{\text{prop1}}$, and span a space $\cS_{\text{prop1}} \subset \cS_{\text{legal}}$ consisting of all states representing correct propagation (within the legal clockspace) of all single-qubit gates of $V$.  Furthermore, $H_{\text{prop1}}|_{\cS_{\text{legal}}}$ decomposes into $T_2+2$ invariant blocks: one ranging over all clock states from time $1$ to time $M$, corresponding to the idling period and spanned by states of the form $\ket{i}\otimes\ket{\hatt}$ for $\hatt \in \{1,\dots,M\}$; and $T_2+1$ many ranging over clock states representing times between 2-qubit gates (which are all of length $L$), spanned by states of the form $U_t\cdots U_1\ket{i}\otimes\ket{\hatt}$ for $\hatt \in \{lL+M,\dots,(l+1)L+M-1\}$. Within each block, $H_{\text{prop1}}|_{\cS_{\text{legal}}}$ corresponds exactly to the usual Kitaev propagation Hamiltonian (See Eq.~\eqref{eq:propagation_hamiltonian}), except only over single qubit gates, and for a total computation time of $M$ in the first case and $L$ in the second. By the usual arguments (see in particular Claim 2 from~\cite{kempe2006}), all non-zero eigenvalues of such a Hamiltonian ranging over $W$ timesteps are at least $c/W^2$ for some constant $c>0$. This implies that the smallest non-zero eigenvalue of $H_{\text{prop1}}|_{\cS_{\text{legal}}}$ is at least $\min\{c/L^2, c/M^2\} = c/M^2$ by the fact that $M \geq L$. Hence, all eigenvectors of $H_2$ orthogonal to $\cS_{\text{prop1}}$ have eigenvalue at least $J = J_1c/M^2$, and Lemma~\ref{lem:projection} implies that for $J_1 \geq \poly(\|H_1\|M^2) = \poly(n)$, $\lambda_0(H_1+H_2)$ can be lower bounded by $\lambda_0(H_1|_{\cS_{\text{prop1}}})-\frac18$.

The remainder of the proof requires us to lower bound the smallest eigenvalue of
\[
    H_{\text{out}}|_{\cS_{\text{prop1}}} + J_{\text{in}} H_{\text{in}}|_{\cS_{\text{prop1}}} + J_2 H_{\text{prop2}}|_{\cS_{\text{prop1}}}\,.
\]
In~\cite{kempe2006} it is shown that the Hamiltonian $(H_{\text{time},t} + H_{\text{qubit},t})|_{\cS_{\text{prop1}}}$ is the same as the restriction to $\cS_{\text{prop1}}$ of 
\begin{eqnarray*}
    \proj{00}_{f_t,s_t} &\otimes& 2\left(\ket{\hatt{-1}} - \ket{\hatt}\right)\left(\bra{\hatt{-1}} - \bra{\hatt} \right) + \\
    \proj{01}_{f_t,s_t} &\otimes& \frac12\left(\ket{\hatt{-1}} - \ket{\hatt}\right)\left(\bra{\hatt{-1}} - \bra{\hatt} \right) + \\
    \proj{10}_{f_t,s_t} &\otimes& \frac12\left(\ket{\hatt{-1}} - \ket{\hatt}\right)\left(\bra{\hatt{-1}} - \bra{\hatt} \right) + \\
    \proj{11}_{f_t,s_t} &\otimes& \left(\ket{\hatt{-1}} + \ket{\hatt}\right)\left(\bra{\hatt{-1}} + \bra{\hatt} \right) \,.
\end{eqnarray*}
Given this equivalent  form of the Hamiltonian (within $\cS_{\text{prop1}}$), it is now easy to see that $\braket{\eta}{H_{\text{prop2}}|\eta} = 0$: $\ket{\eta} \in \cS_{\text{prop1}}$ since it encodes valid propagation for all time-steps, and then by using the fact that at time $t$ $\ket{\eta}$ encodes (correct) propagation via a $CZ$ gate, the above Hamiltonian has $\ket{\eta}$ as a 0-eigenvector. Moreover -- and importantly for us -- this shows that, within the space $\cS_{\text{prop1}}$, $H_{\text{prop2}}$ is positive semi-definite. We will use this fact later to prove bounds on the second smallest eigenvalue of $H$.

\paragraph{Restriction to $\cS_{\text{prop}}$} Returning to the main proof, to bound the smallest eigenvalue of 
\[
    H_{\text{out}}|_{\cS_{\text{prop1}}} + J_{\text{in}} H_{\text{in}}|_{\cS_{\text{prop1}}} + J_2 H_{\text{prop2}}|_{\cS_{\text{prop1}}}\,,
\]
it is sufficient to lower bound instead the Hamiltonian
\[
    H_{\text{out}}|_{\cS_{\text{prop1}}} + J_{\text{in}} H_{\text{in}}|_{\cS_{\text{prop1}}} + H'\,,
\]
where $H'$ is some Hamiltonian acting on $\cS_{\text{prop1}}$ such that $H_{\text{prop2}}|_{\cS_{\text{prop1}}} - H'$ is positive semi-definite. In~\cite{kempe2006}, the authors construct such an $H'$, such that its nullspace $\cS_{\text{prop}}$ is spanned by the states
\[
    \ket{\eta_i} = \frac{1}{\sqrt{T_2+1}} \sum_{l=0}^{T_2} \ket{\eta_{l,i}} = \frac{1}{\sqrt{T+M}} \sum_{t=1}^{T+M} U_t\cdots U_1\ket{i} \otimes \ket{\hatt}\,,
\]
and whose smallest non-zero eigenvalue is $\frac{c}{LT_2^2} \geq \frac{c}{T^2}$. For our proof we can take precisely the same $H'$ (since our $H_{\text{prop2}}$ is identical to theirs), and therefore omit details of the construction here, instead referring to the proof given in~\cite{kempe2006}. 

We now apply Lemma~\ref{lem:projection} once more, this time with
\[
    H_1 = (H_{\text{out}} + J_{\text{in}}H_{\text{in}})|_{\cS_{\text{prop1}}} \qquad H_2 = J_2 H'
    ,.
\]
The eigenvectors of $H_2$ orthogonal to $\cS_{\text{prop}}$ have eigenvalues at least $J = J_2c/T^2$. One can therefore choose $J_2 \geq \poly(\|H_1\|T^2) = \poly(n)$ such that the smallest eigenvalue of $H_1 + H_2$ is lower bounded by $\lambda_0(H_1|_{\cS_{\text{prop}}}) - \frac18$, and it therefore suffices to consider at last the Hamiltonian
\[
    H_{\text{out}}|_{\cS_{\text{prop}}} + J_{\text{in}}H_{\text{in}}|_{\cS_{\text{prop}}}\,.
\]

\paragraph{Restriction to $\cS_{\text{in}}$} The final step of the proof is to apply Lemma~\ref{lem:projection} one more time with
\[
    H_1 = H_{\text{out}}|_{\cS_{\text{prop}}} \qquad H_2 = J_{\text{in}} H_{\text{in}}|_{\cS_{\text{prop}}}\,.
\]
By our construction (in particular, the lack of any witness register), the intersection of the nullspace of $H_{\text{in}}$ with $\cS_{\text{prop}}$ is 1-dimensional, and consists of the single history state 
\begin{equation}\label{eq:single_history_state}
    \ket{\eta} = \frac{1}{\sqrt{T+M}} \sum_{t=1}^{T+M} U_t\cdots U_1 \ket{x,0} \otimes \ket{\hat{t}}\,.
\end{equation}
This state is a 0-eigenvector of $H_2$. Any other eigenstate orthogonal to it (but inside $\cS_{\text{prop}}$) has eigenvalue at least $\frac{J_{\text{in}}}{T+M}$.\footnote{Any orthogonal eigenstate \emph{outside} $\cS_{\text{prop}}$ must have energy (with respect to the full Hamiltonian $H$) much larger than this via the previous applications of Lemma~\ref{lem:projection}.} Hence, choosing $J_{\text{in}} \geq \poly(\|H_1\|(T+M)) = \poly(n)$ is enough to ensure that the smallest eigenvalue of $(H_1 + H_2)$ is lower bounded by $\lambda_0(H_{\text{out}}|_{\cS_{\text{in}}})-\frac18$. 

Since the space $\cS_{\text{in}}$ is precisely the state $\ket{\eta}$, then we have
\[
    H_{\text{out}}|_{\cS_{\text{in}}} = \langle x,0 | U_1^\dagger\cdots U_{T+M}^\dagger (\proj{0}_0 \otimes I)U_{T+M}^\dagger\cdots U_1 | x,0\rangle = \epsilon\,,
\]
which immediately gives a lower bound of $\epsilon - \frac18$ on $\lambda_0(H_1 + H_2)$ and proves its uniqueness. Putting everything together, we get a lower bound on the smallest eigenvalue of the full Hamiltonian $H$ of $\lambda_0(H) \geq \epsilon - \frac12$. This can actually be improved by noting that, for the final two  applications of Lemma~\ref{lem:projection}, both Hamiltonians $H_1$ and $H_2$ are positive semi-definite, and hence we can obtain the slightly tighter bound of $\lambda_0(H) \geq \epsilon - \frac14$. 
\end{proof}

\subsection{Proof of Lemma~\ref{lem:fidelity_pert}}\label{app:proofs_2}
We will first introduce (review) some notation and concepts, following the conventions of~\cite{kempe2006}. 
\begin{definition} Let $\lambda^{*}\in \mathbb{R}$ be some cut-off, and let $\mathcal{H} = \mathcal{L}_{-} \oplus \mathcal{L}_{+}$, where $\mathcal{L}_{+}$ is the space spanned by eigenvectors of $H$ with eigenvalues $\lambda \geq \lambda_{*}$ and $\mathcal{L}_{-}$ is spanned by eigenvectors of $H$ of eigenvalue $\lambda<\lambda_{*}$. Let $\Pi_{\pm}$ be the corresponding projection onto $\mathcal{L}_{\pm}$. For an operator $X$ on $\mathcal{H}$ we define $X_{++}=X|_{\mathcal{L}_{+}} = \Pi_{+} X\Pi_{+}$, $X_{--}=X|_{\mathcal{L}_{-}} = \Pi_{-} X\Pi_{-}$, $X_{+-}= \Pi_{+} X\Pi_{-}$ and $X_{-+}= \Pi_{-} X\Pi_{+}$ .
\end{definition}
Let $\tilde{H} = Q + P$ be a sum of two Hamiltonians $Q$ and $P$, referred to as the \textit{unperturbed} and \textit{perturbation} Hamiltonian, respectively. Write $\lambda_j(Q),\ket{\psi_j}$ for the $j$th eigenvalue and eigenvector of $Q$, and denote $\lambda_j(\tilde{H}),\tilde{\ket{\psi_j}}$ for the $j$th eigenvalue and eigenvector of $\tilde{H}$. The \textit{resolvent} of $\tilde{H}$ is defined as 
\begin{align}
    \tilde{G}(z) := (zI-\tilde{H})^{-1} = \sum_{j} \left(z-\lambda_j(\tilde{H})\right)^{-1} \tilde{\ket{\psi_j}} \tilde{\bra{\psi_j}}.
\end{align}
Let $\lambda_{*}\in \mathbb{R}$ be some cut-off on the spectrum of $Q$. Define the self-energy as
\begin{align}
    \Sigma_{-} (z) := zI_{-} - \tilde{G}^{-1}_{--}(z).
\end{align}
We now have all definitions available needed for the proof of Lemma~\ref{lem:fidelity_pert}:
\lemmaNine*
\begin{proof}[Proof of Lemma~\ref{lem:fidelity_pert}]
We follow the proof of Lemma 11 in Ref.~\cite{kempe2006}, now for arbitrary eigenvectors and with slightly improved bounds on the fidelity. We will repeatedly use that by Lemma~\ref{lem:delta_close}, we have that $\lambda_i(\tilde{H}) \leq \lambda_i(H_\text{eff}) + c_r\mu$ via our choice of $\mu$. Let $\ket{\tilde{v}_{i,-}} = \Pi_{-} \ket{\tilde{v}_i}$. We must have that
\begin{align*}
    \norm{\Pi_{+} \tilde{H} \ket{\tilde{v}_i}} &= \tilde{\lambda}_i \norm{\Pi_{+} \ket{\tilde{v}_i}}\\
    &\leq (\lambda_i(H_\text{eff})+c_r\mu) \norm{\Pi_{+} \ket{\tilde{v}_i}},
\end{align*}
as well as
\begin{align*}
    \norm{\Pi_{+} \tilde{H} \ket{\tilde{v}_i}} &= \norm{\Pi_{+} Q \ket{\tilde{v}_i} + \Pi_{+} P \ket{\tilde{v}_i}}\\
    &\geq \norm{\Pi_{+} Q \ket{\tilde{v}_i}} - \norm{P}\\
    &\geq \lambda_{+} \norm{\Pi_{+} \ket{\tilde{v}_i}} - \norm{P}.
\end{align*}
Combining both, we obtain
\begin{align*}
    \norm{\Pi_{+} \ket{\tilde{v}_i}} \leq \frac{\norm{P}}{\lambda_+ - \lambda_i(H_\text{eff})-c_r\mu}.
\end{align*}
Therefore, we have that 
\begin{align}
\abs{\bra{\tilde{v}_i}\ket{\tilde{v}_{i,-}}}^2 &= \norm{\Pi_{-} \ket{\tilde{v}_i}}^2 \nonumber\\ 
&= 1-\norm{\Pi_{+} \ket{\tilde{v}_i}}^2
     \nonumber\\
    &\geq 1-\left(\frac{\norm{P}}{c_r \mu^{-3}-\lambda_i(H_\text{eff})-c_r\mu}\right)^2
\label{eq:fidelity_ineq1}
\end{align}
As in Ref.~\cite{kempe2006}, we will now move our efforts to bounding the fidelity between $\ket{\tilde{v}_{i,{-}}}$ and $\ket{v_{\text{eff},i}}$. We have that
\begin{align*}
    \tilde{G}_{--} &= \sum_{i} \left(z-\lambda_i(\tilde{H})\right)^{-1} \Pi_{-} \ket{\tilde{v}_{i}} \bra{\tilde{v}_{i}} \Pi_{-}\\
    &=\sum_{i} \left(z-\lambda_i(\tilde{H})\right)^{-1} \ket{\tilde{v}_{i,-}} \bra{\tilde{v}_{i,-}}.
\end{align*}
For the self-energy we have then
\begin{align*}
    \Sigma_{-}(z) &= z I_{-} - \tilde{G}_{--}^{-1}\\
    &= z \sum_{i} \ket{\tilde{v}_{i,-}} \bra{\tilde{v}_{i,-}} - \sum_{i} (z-\lambda_i(\tilde{H}))  \ket{\tilde{v}_{i,-}} \bra{\tilde{v}_{i,-}}\\
    & =  \sum_{i} \lambda_i(\tilde{H}) \ket{\tilde{v}_{i,-}} \bra{\tilde{v}_{i,-}}.
\end{align*}
Hence, $\ket{\tilde{v}_{i,-}}$ is an eigenstate of $\Sigma_{-}(z)$ with eigenvalue $\lambda_i(\tilde{H})$. By our assumptions we have that
\begin{align*}
    c_r\mu &\geq \norm{\Sigma_{-}(\lambda_i(\tilde{H})) - H_\text{eff}}\\
    &\geq  \bra{\tilde{v}_{i,-}} ( \Sigma_{-} (z) - H_\text{eff})  \ket{\tilde{v}_{i,-}} \\
    &=  \lambda_i(\tilde{H}) - \bra{\tilde{v}_{i,-}} H_\text{eff}  \ket{\tilde{v}_{i,-}},
\end{align*}
and so
\begin{align*}
    \bra{\tilde{v}_{i,-}} H_\text{eff}  \ket{\tilde{v}_{i,-}} &\leq \lambda_i(\tilde{H}) + c_r\mu\\ &\leq \lambda_i(H_\text{eff})+2c_r\mu.
\end{align*}
Let us write $ \ket{\tilde{v}_{i,-}} = a \ket{v_{\text{eff},i}} + b \ket{v_{\text{eff},i}^\perp}$, where $a,b\in\mathbb{R}$ and $a^2+b^2=1$. We have that $a={\bra{\tilde{v}_{i,-}} \ket{v_{\text{eff},i}}}$ and $b={\bra{\tilde{v}_{i,-}} \ket{v^{\perp}_{\text{eff},i}}}$. We obtain 
\begin{align*}
    \bra{\tilde{v}_{i,-}} H_\text{eff}  \ket{\tilde{v}_{i,-}} &= (a \bra{v_{\text{eff},i}} + b \bra{v_{\text{eff},i}^\perp}) H_\text{eff}  (a\ket{v_{\text{eff},i}} + b \ket{v_{\text{eff},i}^\perp})\\
    &= a^2 \bra{v_{\text{eff},i}} H_\text{eff} \ket{v_{\text{eff},i}} + b^2 \bra{v^{\perp}_{\text{eff},i}} H_\text{eff} \ket{v^{\perp}_{\text{eff},i}}\\
    &\geq a^2 \lambda_i(H_\text{eff}) + (1-a^2) \left(\lambda_i(H_\text{eff})+ \gamma_i(H_\text{eff})\right)\\
    &=  \lambda_i(H_\text{eff}) + (1-a^2) \gamma_i(H_\text{eff}).
\end{align*}
Again by combining the two inequalities, we find that
\begin{align}
    \abs{\bra{\tilde{v}_{i,-} }\ket{v_{\text{eff},i}}}^2 = a^2 \geq 1-\frac{2 c_r\mu}{\gamma_i (H_\text{eff})}.
    \label{eq:fidelity_ineq2}
\end{align}
Combining both Eq.~\eqref{eq:fidelity_ineq1} and~\eqref{eq:fidelity_ineq2}, we have that
\begin{align*}
    \abs{\bra{\tilde{v}_i}\ket{v_{\text{eff},i}}}^2 &\geq 1-\left(\sqrt{1-\abs{\bra{\tilde{v}_i}\ket{\tilde{v}_{i,-}}}^2}  +\sqrt{1-\abs{\bra{\tilde{v}_{i,-} }\ket{v_{\text{eff},i}}}^2}\right)^2\\
    &=1-\left(\frac{\norm{P}}{c_r \mu^{-3}-\lambda_i(H_\text{eff})-c_r\mu} + \sqrt{\frac{2 c_r\mu}{\gamma_i (H_\text{eff})}} \right)^2,
\end{align*}
which completes the proof.
\qedhere
\end{proof}